\documentclass[letterpaper,11pt]{article}

\usepackage{amsmath, amsthm, amssymb}

\usepackage{amsfonts}
\usepackage{graphicx}
\usepackage{epsfig}
\usepackage{color}
\usepackage{algorithm}
\usepackage[noend]{algorithmic}
\usepackage{fullpage}
\usepackage{graphics}

\usepackage{verbatim}

\newtheorem{theorem}{Theorem}

\newtheorem{lemma}[theorem]{Lemma}

\newtheorem{claim}[theorem]{Claim}

\newtheorem{observation}[theorem]{Observation}

\newcommand{\cover}{\mbox{coverage}}
\newcommand{\dur}{\mbox{M}}
\newcommand{\Dur}{\mbox{M}}
\newcommand{\inst}{\textbf{S}}

\newcommand{\W}{\mathcal W}
\newcommand{\C}{\mathcal C}

\newcommand{\head}{h}
\newcommand{\tail}{\tau}
\newcommand{\load}{\mbox{load}}

{\makeatletter
 \gdef\xxxmark{%
   \expandafter\ifx\csname @mpargs\endcsname\relax % in minipage?
     \expandafter\ifx\csname @captype\endcsname\relax % in figure/caption?
       \marginpar{xxx}% not in a caption or minipage, can use marginpar
     \else
       xxx % notice trailing space
     \fi
   \else
     xxx % notice trailing space
   \fi}
 \gdef\xxx{\@ifnextchar[\xxx@lab\xxx@nolab}
 \long\gdef\xxx@lab[#1]#2{{\bf [\xxxmark #2 ---{\sc #1}]}}
 \long\gdef\xxx@nolab#1{{\bf [\xxxmark #1]}}
 \gdef\turnoffxxx{\long\gdef\xxx@lab[##1]##2{}\long\gdef\xxx@nolab##1{}}%
} 

\title{Decomposing Coverings and the Planar Sensor Cover Problem}
\author{Matt Gibson \and Kasturi Varadarajan}

\begin{document}

\bibliographystyle{plain}
\maketitle

\begin{abstract}
We show  that a $k$-fold covering using translates of an arbitrary convex polygon can 
be decomposed into $\Omega(k)$ covers (using an efficient algorithm). We generalize this result 
to obtain a constant factor approximation to the sensor cover problem where the ranges of
the sensors are translates of a given convex polygon. The crucial ingredient in this
generalization is a constant factor
approximation algorithm for a one-dimensional version of the sensor
cover problem, called the \textit{Restricted Strip Cover} (RSC) problem, where
sensors are intervals of possibly different lengths. Our algorithm for RSC 
improves on the previous $O(\log \log \log n)$ approximation. 
\end{abstract}

\section{Introduction}
\label{sec:intro}
Let us call an object (set) $P$ in the plane {\em cover-decomposable} if there there exists
a constant $c$ (which may depend on $P$) such that any collection of translates of $P$, with
the property that every point in the plane has $c$ or more translates covering it, can be
partitioned into two covers. Pach conjectured in the 1980s that every convex object is cover 
decomposable \cite{pach86, pt07}, and this remains open. Let us focus on a  finite version
of this definition and say that $P$ is cover-decomposable if there exists a constant $c$ such
that any finite collection of translates of $P$ can be partitioned into two sub-collections,
so that each sub-collection covers every point in the plane covered by $c$ or more translates
in the original collection. 

In the 1980's, Mani and Pach \cite{mp88} showed that a unit disk is cover-decomposable (with
the constant $c$ being $33$.)  Also in the 1980's, Pach \cite{pach86} showed that any centrally 
symmetric convex polygon is cover-decomposable. Tardos and T{\'o}th \cite{tt07} showed somewhat
more recently that any triangle is cover-decomposable. Finally, a very recent result due to
P\'{a}lv\"{o}lgyi  and T{\'o}th  \cite{pt08} shows that any convex polygon is cover-decomposable. The 
constant $c$ in the results of \cite{pach86} and \cite{pt08} depends on the
convex polygon, in particular the number of its sides, and that is why these results say nothing about
the original conjecture of Pach. Examples of non-convex polygons that are not 
cover-decomposable are known \cite{p09}.

Motivated partly by questions in scheduling sensors \cite{bej07}, an extension of the 
cover-decomposability question has recently attracted a lot of attention: Given a collection 
of translates of $P$ and any integer $k$, partition the collection into as many 
sub-collections as possible
so that each sub-collection covers every point covered by $k$ or more of the original 
translates. While the original results on cover-decomposability do yield non-trivial 
bounds for this question, these are usually far from optimal. For instance,
Tardos and T{\'o}th \cite{tt07} implies that a $k$-fold cover with translates of a triangle
can be partitioned into $\Omega(\log k)$ covers. 

In this line of work, Pach and T{\'o}th \cite{pt07} showed that a $k$-fold cover with
a centrally symmetric convex polygon $P$  can be decomposed into $\Omega(\sqrt{k})$ covers,
where the constant as before depends on $P$. Aloupis et al. \cite{acclor08} improved this
result and obtained an optimal bound, showing that one can obtain $\Omega(k)$ covers. 
Recently, Gibson \cite{g09} was able to show an optimal $\Omega(k)$ bound when $P$ is
a triangle. All these results have corresponding efficient algorithms that compute the
desired decompositions. Aloupis et al. \cite{accls08} consider other related problems.

The problem of decomposing multiple coverings seems to be harder if instead of a convex
polygon we have a unit disk. Pandit, Pemmaraju and Varadarajan \cite{pv08} consider a special case 
where the universe that needs to be covered is the same as the centers of the
covering disks. For this version of the problem, better known as
the domatic partition problem for {\em unit disk graphs} \cite{pp06}, they 
show that it is possible to compute $\Omega(L)$ disjoint covers in polynomial time.

\subsection{Sensor Cover}
The work on decomposing multiple coverings is related to a problem motivated by
scheduling sensors. Suppose we have a universe, which is simply some collection of points, and a set of sensors such that each sensor covers some subset of the universe.  Further suppose that each sensor is powered by a battery and thus can only be turned
on for some amount of time.  We refer to this amount of time as the sensor's \textit{duration}.  We are interested in scheduling a start time to each of the
sensors such that the entire universe is covered for as long as possible.  This problem was introduced by Buchsbaum et al. \cite{bej07} as the 
\textit{sensor cover problem}.  We only consider the non-preemptive case in 
which once a sensor has been turned on, it will remain on until its duration 
has been depleted.

We now formally define the combinatorial sensor cover problem, followed
by the geometric instances that are the subject of this paper. We are given a 
finite universe $U$ that we wish to cover, and a set $\inst$ of $n$
sensors.  For each sensor $s \in \inst$, we let $R(s) \subseteq U$ denote the 
region that $s$ covers.  We call this region the \textit{range} of $s$.  For each $x \in R(s)$, we say that $s$ is 
\textit{live} at $x$.  Each sensor $s$ also has a duration $d(s)$ which is 
some positive integer.

A schedule of the set $\inst$ of sensors is an assignment of a positive 
integer, called the start time, to each sensor in some subset 
$S' \subseteq \inst$. We will denote by $t(s)$ the
start time of sensor $s$. The sensors in $\inst \setminus S'$ are said to be 
unassigned. A sensor $s$ that is assigned a start time $t(s)$ is said to be 
\textit{active} at times $\{t(s), t(s) +
1, . . . , t(s) + d(s) - 1\}$.

Let $S$ be some schedule of $\inst$.  A point $x \in U$ is said to be covered 
at time $t > 0$ if there is a sensor $s$ such that $x \in R(s)$ and $s$ is 
active at time $t$.  For each $x \in U$, define the duration of $x$ in the 
schedule to be $M(S, x) = max\{j : \forall j' \leq j, \exists s \in S, s \text{ covers } x \text{ at time } j'\}$. 
(If no sensor covers $x$ at time 1, then define $M(S, x) = 0$.) The 
duration of the schedule $S$ is defined to be $M(S) = \text{min}_x M(S, x)$.  
The goal of the problem is to compute a schedule of maximum duration.

The \textit{load} at a point $x \in U$ is 
$L(x) = \sum_{s \in \inst: x \in R(s)} d(s)$.  The \textit{load} of the 
problem instance is $L = \text{min}_x L(x)$.  Let $OPT$ denote the duration of 
an optimal schedule.  Clearly, $OPT \leq L$, and thus any
approximation ratio that is with respect to $L$ is also with respect to $OPT$.

A closely related problem to the sensor cover problem is the 
\textit{domatic partition problem}.  In this problem, we are given a 
graph with the goal of finding the maximum number of disjoint dominating sets.
A \textit{dominating set} is a subset of the vertices such that for each 
vertex in the graph, either it is in the set or it has a neighbor in the set.  
Domatic partition can be viewed as a special case of the sensor cover problem 
where the universe is the vertex set, each vertex of the graph is a sensor, 
the range of each sensor is its corresponding 
vertex's closed neighborhood, and each sensor has unit duration.  
Feige et al. \cite{fhk03} show that it is NP-hard to approximate this 
problem to within a $\log n$-factor and give a simple randomized algorithm 
that achieves an $O(\log n)$-approximation, where $n$ is the number of
vertices in the graph.

As pointed by Buchsbaum et al. \cite{bej07}, the lower bound above given by 
Feige et al. implies that general sensor cover cannot be approximated to 
better than a $\log n$ factor.  On the positive side, 
Buchsbaum et al. \cite{bej07} present a poly-time algorithm for the sensor
cover problem that returns an $O(\log U)$ approximation. This algorithm 
extends an algorithm for the set cover packing problem \cite{fhk03}, which is 
the special case of the sensor cover problem with the duration of all sensors 
being $1$. In many applications, the sensors do not cover arbitrary subsets 
of the universe, but rather the points in the universe lie in some geometric 
space and the sensors cover some geometric subset of the universe.  
In such cases, we will see that it is possible to do better than the 
$\log n$ lower bound for general sensor cover.

\paragraph{Restricted Strip Cover (RSC).} Here, the universe $U$ is a set of points
on the real line, and the range $R(s)$ of each sensor $s$ is an interval
on the real line. The RSC was introduced and studied 
by Buchsbaum et al. \cite{bej07}, who    
showed that it is NP-hard and give a polynomial-time 
$O(\log \log \log n)$-approximation algorithm, where $n$ is the number of 
sensors.  They show that their algorithm does better for special cases 
of RSC.  In particular, they show that their algorithm is a 
$(2 + \epsilon)$-approximation for any $\epsilon > 0$ when the sensors are 
non-nested; this includes the case where the ranges of all the sensors
have the same size. They also give an example instance whose load
is $4$ but the duration of the optimal schedule is $3$.  

The RSC bears some resemblance to the well studied 
{\em dynamic storage allocation} \cite{g99, bkk03}. The RSC problem can 
be viewed in the following way. We are given a set of rectangles, and we are 
allowed to slide each rectangle vertically; the goal is to find a placement 
of the rectangles so that we cover a horizontal strip that is as wide as 
possible. In the dynamic storage allocation problem, we are also given a set
of rectangles, each of which we are allowed to slide vertically; the goal
is to find a placement of the rectangles such that no two of them
overlap and the rectangles are contained in a horizontal strip that
is as thin as possible. The dynamic storage allocation problem admits
constant factor approximation algorithms \cite{g99, bkk03}, and these
are with respect to the load, where now the load is the {\em maximum}
of the pointwise loads. We refer the reader to \cite{bej07} for a 
review of the similarities of the RSC to other problems 
%\cite{agp04,dvc04, dkmt08, gj90, kr00, ph03, sp01} 
studied in the literature.

\paragraph{Planar Sensor Cover for Polygon Translates} Here, the 
universe $U$ is a set of points in $\Re^2$, and the range $R(s)$ of each sensor is a translate
of a fixed convex polygon. In the remainder of this paper, we will
refer to this problem as simply the planar sensor cover problem.
The results on decomposing multiple coverings have a clear implication 
for the special case of the planar sensor cover problem where all
sensors have unit durations. For instance, the result of Aloupis 
\cite{acclor08} gives a constant factor approximation for centrally symmetric
polygons, since it decomposes an $L$-fold cover into $\Omega(L)$ covers,
thus yielding a schedule of duration $\Omega(L)$. Gibson's result \cite{g09}
gives a constant factor approximation for triangles.

For the planar sensor cover problem with the durations of the sensors
not being the same, the best known result is the logarithmic approximation 
inherited from the combinatorial sensor cover problem. An exception is
the recent constant factor approximation algorithm of Gibson \cite{g09} 
for triangles, which is based on a preliminary version of this paper, and
uses our results on RSC. 

\subsection{Our Contributions}

\paragraph{Decomposing Multiple Coverings.} We obtain an optimal result
for translates of an arbitrary convex polygon:

\begin{theorem}
\label{thm:decompose-main}
For any convex polygon $P$ in the plane, there exists a constant $\alpha \geq 1$ so that
for any $k \geq 1$ and any finite collection of translates of $P$, we can partition the 
collection into $k/\alpha$ sub-collections, each of which covers any point in the plane that is
covered by $k$ or more translates in the original collection. Such a partition can be
computed by an efficient algorithm.
\end{theorem}  

Our techniques build upon the recent work of Aloupis et al. \cite{acclor08} for centrally
symmetric convex polygons. (A polygon is centrally symmetric with respect to the origin
if whenever it contains point $p$ it also contains $-p$.) A key idea of theirs is to
focus on the level curves corresponding to the wedges at the vertices of $P$. The interaction
of these level curves can be complex, but they show that is sufficient to work within
a region where the interaction is much more controlled. It is only for centrally
symmetric convex polygons that they establish such nice properties of the
interaction. The notion of level curves is also
central to our work, but the main point of departure is the simplicity of the way in
which we handle the level curve interactions. One other important idea we need is 
a generalization of the proof technique of Gibson \cite{g09} for the case of triangles.

\paragraph{Restricted Strip Cover.}
We improve upon the $O(\log \log \log n)$ approximation of $\cite{bej07}$ and
give the first constant factor approximation (a ratio of $5$) for RSC.  The
work of $\cite{bej07}$ starts off with the observation that if all the sensors
have unit duration then it is possible to compute a schedule whose duration
is equal to the load of the instance. The case of non-uniform duration is
handled by reduction to several instances of the uniform duration case. 
The tool used for this is a technique called {\em grouping} where several
sensors of small duration are combined to form one sensor of large
duration. The question of how the groups are to be formed is addressed
in a clever way, but the reduction entails a non-constant loss in the load and 
hence the $O(\log \log \log n)$ approximation factor.

We take a different and conceptually simpler approach here. Our algorithm
is greedy and schedules sensors one by one. The scheduling rule manages to
ensure that we do not have more than 5 sensors overlapping any particular 
point at one time. Hence we obtain a schedule whose duration is at least
a fifth of the load. One idea that the scheduling rule uses is that
if there are two sensors $s$ and $s'$ such that $R(s)$ is strictly contained
in  $R(s')$, we schedule $s'$ before we schedule $s$. Another idea is
to consider the duration of the sensors in an indirect way -- for the
next sensor to be scheduled, the durations of the unscheduled sensors
is irrelevant but only their ranges; however the durations of the already
scheduled sensors does play a crucial rule.  Since our algorithm is greedy
it has a simple implementation with a reasonable running time. We
have not attempted to optimize the factor of $5$ that our analysis
guarantees.

\paragraph{Planar Sensor Cover.}
We give a constant factor approximation for the planar sensor cover problem,
where the range of each sensor is a translate of a convex polygon, improving upon 
the previous best logarithmic factor. Essentially, we show that 
we can obtain a constant factor approximation by invoking several instances
of the RSC, one for each vertex of the convex polygon, which we solve using our 
$5$-approximation. Our greedy algorithm for RSC turns out to be exactly what is needed
to generalize the result of Theorem \ref{thm:decompose-main} to the case of
non-uniform durations.

\paragraph{Organization of the Paper.}
In Section \ref{sec:prelims}, we recall crucial tools from previous work on the
problem of decomposing multiple coverings. In Section \ref{sec:decompose}, we prove
Theorem \ref{thm:decompose-main}. In Section \ref{sec:RSC}, we present our
constant factor approximation for RSC and obtain as a consequence the results
on planar sensor cover.

\section{Decomposing Coverings: Preliminaries}
\label{sec:prelims}
It is convenient to prove Theorem \ref{thm:decompose-main} in its dual form as done in 
\cite{tt07,acclor08,g09}.  Suppose we are given a polygon $P$.  Fix $O$, the centroid of $P$, as the origin in the plane.  For a planar set $T$ and a point $x$ in the plane, let $T(x)$ denote the translate of $T$ with centroid $x$.  Let $\bar{P}$ be the reflection through $O$ of the polygon $P$.  For points $p$ and $x$ in the plane, $p \in P(x)$ if and only if $x \in \bar{P}(p)$.

Because of this transformation, it is sufficient for us to show that there exists a
constant $\alpha \geq 1$ so that for any $k \geq 1$ and any collection $Q$ of points in the
plane, it is possible to assign each point in $Q$ a color from 
$\{1, 2, \ldots, \frac{k}{\alpha} \}$, so that any translate of $\bar{P}$ with
$|\bar{P} \cap Q| \geq k$ contains a point colored $i$, for each $1 \leq i \leq \frac{k}{\alpha}$.

\paragraph{Polygons to Wedges.}
Denote the vertices of $\bar{P}$ to be $p_0, p_1, p_2, \ldots, p_{n-1}$ in counterclockwise 
order.  Addition and subtraction of indices of these vertices will be taken modulo $n$ throughout the paper.  The set of indices between index $i$ and index $j$ in counterclockwise order are denoted $[i,j]$. We now transform the problem further, so that instead of dealing with translates
of $\bar{P}$, we can deal with translates of the $n$ {\em wedges} corresponding to the vertices
of $\bar{P}$ \cite{pt07,tt07,acclor08,g09}. 

Let $c$ be equal to half the minimum distance
between two points on non-consecutive edges of $\bar{P}$. We lay a square grid of side $c$ on 
the plane; any translate of $\bar{P}$ intersects $\beta \in O(1)$ grid cells, and each grid cell 
intersects at most two sides of a translate; moreover, if a grid cell does intersect two sides
of a translate, then these sides must be adjacent in $\bar{P}$   

For a subset (region) $R$ of the plane and for a subset $X$ of points, denote $\load_X(R)$ to be 
the number of points in $X$ that lie in $R$.  We call this value the load of region $R$ with respect to $X$.  Since each translate $\bar{P}(u)$ intersects at most $\beta$ grid cells, 
$\bar{P}(u)$ must contain load at least $k/\beta$ within some grid cell if its load with respect
to $Q$ is at least $k$. We can therefore make the points of $Q$ within such a grid cell 
``responsible'' for $\bar{P}(u)$.

Since each grid cell intersects at most two edges of $\bar{P}(u)$, it must be that the 
intersection of a grid cell with $\bar{P}(u)$ is the same as the intersection of the grid cell 
with a wedge whose boundaries are parallel to two adjacent edges of $\bar{P}(u)$.  If one 
boundary of the wedge is parallel to the edge $p_{i-1}p_i$ of $\bar{P}$ and the other is parallel with $p_ip_{i+1}$ of $\bar{P}$, then we call the wedge an $i$-wedge.  For a point $q$ in the 
plane, we denote $W_i(q)$ to be the $i$-wedge with apex $q$.  See Figure \ref{fig:wedge} for an illustration.

\begin{figure}[htpb]
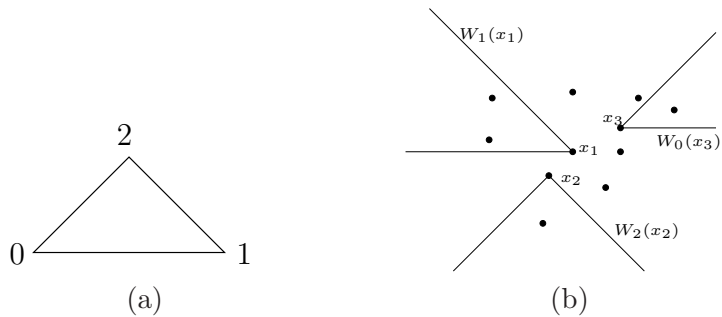

\centering
\begin{tabular}{c@{\hspace{0.1\linewidth}}c}
\input{tri.pstex_t} &
\input{wedges.pstex_t} \\
(a) & (b)
\end{tabular}
\caption{(a) Suppose this triangle is our polygon with vertices indexed accordingly.
(b) A 0-wedge, 1-wedge, and 2-wedge with respect to the polygon.}
\label{fig:wedge}
\end{figure}

Because of these observations, Theorem \ref{thm:decompose-main} is established by applying the 
following theorem to the points $Y$ within each grid cell $G$.  

\begin{theorem}
\label{thm:decompose-wedges}
There exists a constant $\alpha' \geq 1$ so that for any $k \geq 1$ and any collection $Y$ of 
points in the plane, it is possible to assign each point in $Y$ a color from 
$\{1, 2, \ldots, \frac{k}{\alpha'} \}$, so that any $i$-wedge that contains $k$ or
more points from $Y$ contains a point colored $j$, for each $1 \leq j \leq \frac{k}{\alpha}$.
\end{theorem}

We prove Theorem \ref{thm:decompose-wedges} in Section \ref{sec:decompose}.
We assume that the point set $Y$ is in general position -- a line
parallel to a side of $\bar{P}$ contains at most one point in $Y$. It is
straightforward to perturb the input to the original problem so
that this assumption holds for $Y$.

\paragraph{Level Curves.}
We will now define a boundary for an $i \in \{0, 1, \ldots n - 1\}$ and  
positive integer $r$. This boundary has the property that any $i$-wedge placed 
on or ``inside'' the boundary has load at least $r$, and any $i$-wedge 
placed ``outside'' the boundary has load less than $r$.  That is, the number of points in $W_i(x) \cap Y$ for any $x$ inside the boundary or on the 
boundary is at least $r$ and is less than $r$ for any $x$ outside the 
boundary.  This boundary is called a \textit{level curve} \cite{acclor08,g09} and extends the 
definition of \textit{boundary points} \cite{pach86, pt07}.  Let $\W_i^j$ be 
the set of apices of all $i$-wedges $W$ such that $\load_{Y} (W) = j$.  
For each $i = 0, 1, \ldots n-1$, let 
the level curve $\C_i(r)$ be the boundary of the region 
$\W_i^{\geq r} = \bigcup_{j \geq r} \W_i^j$ for each $i = 0, 1, \ldots n-1$.

Note that $\C_i(r)$ is a monotone staircase polygonal path with edges that are parallel to the edges of an $i$-wedge.  See Figure \ref{fig:m1}. We have the 
following observations:

\begin{figure}
\begin{center}
\input{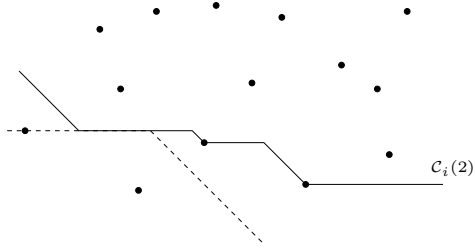}
\caption{An example of a level curve $\C_i(r)$ for $r = 2$.  Note that any 
$i$-wedge with apex on $\C_i(2)$ (e.g. the dotted wedge) contains load at 
least $2$.}
\label{fig:m1}
\end{center}
\end{figure}

\begin{observation}
\label{observ:load}
 For any $x \in \C_i(r)$, $r \leq \load_{Y}(W_i(x)) 
\leq r + 1$.
\end{observation}

\begin{observation}
\label{observ:2}
 Any $i$-wedge $W$ such that $\load_{Y}(W) \geq r$  contains an $i$-wedge 
whose apex belongs to $\C_i(r)$.
\end{observation}

Observe that one of the two extreme edges of the level curve 
$\C_i(r)$ is a semi-infinite ray parallel to edge $p_{i-1}p_i$.  Let $\head_i$ denote the first point along this ray such that for all points $y$ on the ray that lie after $\head_i$, $W_i(y)\cap Y = W_i(\head_i)\cap Y$.  We call $\head_i$ the {\em head} of $\C_i(r)$. 
(Because of non-degeneracy, the head is simply the origin of the ray.)
The other extreme edge of $\C_i(r)$ is parallel to
edge $p_i p_{i+1}$.  Let $\tail_i$ denote the first point along this ray such that for all points $y$ on the ray that lie after $\tail_i$, $W_i(y)\cap Y = W_i(\tail_i)\cap Y$.  We call $\tail_i$ the
{\em tail} of $\C_i(r)$. See Figure \ref{fig:m2}.

\begin{figure}
\begin{center}
\input{levelCurve2.pstex_t}
\caption{Level curve $\C_i(r)$ with $\head_i$ and $\tail_i$ denoted.}
\label{fig:m2}
\end{center}
\end{figure}

\subsection{Simple Algorithm for One Level Curve}
\label{sec:simple}
Observation \ref{observ:2} implies that is sufficient to prove Theorem \ref{thm:decompose-wedges}
for the $i$-wedges with apex on $\C_i(k)$, for each $0 \leq i \leq n - 1$. In order to
do this, we will need a procedure that takes as input one level curve $\C_i(k)$, a
positive integer $t$, and a subset $Q \subseteq Y$. The input to the procedure has the
guarantee that for any $i$-wedge $W$ with apex on $\C_i(k)$, we have $|W \cap Q| \geq 2t$.
The goal is to output a partial coloring of the points of $Q$ with colors 
$\{1, 2, \ldots, t\}$ so that any $i$-wedge $W$ with apex on $\C_i(k)$ (a) contains
a point colored $j$, for $1 \leq j \leq t$, and (b) contains at most $2t$ colored points.

It is known \cite{acclor08} that such a procedure exists. The reason is that for any $q \in Q$,
the set $I(q) = \{u \in \C_i(k) | q \in W_i(u)\}$ of apexes of $i$-wedges containing $q$  is
an ``interval'' of $\C_i(k)$.  See Figure \ref{fig:interval} for an illustration. We consider
these intervals in an order such that if interval $I$ properly contains interval $I'$, then
we consider $I$ before $I'$. Considering intervals in such an order, we add an interval
into our working set if it covers a point of $\C_i(k)$ that is not covered by previous 
intervals in the working set. Notice that after all intervals have been considered, the 
working set forms a cover of $\C_i(k)$. Now, we repeatedly throw out intervals from the working
set that are redundant -- an interval is redundant if throwing it out of the current working
set does not affect coverage of $\C_i(k)$. 

The final non-redundant working set covers $\C_i(k)$, but also has no more than two intervals
covering any point of $\C_i(k)$. We give the color $1$ to the points in $Q$ that give rise to 
the intervals in our working set. We repeat this process $t - 1$ more times. It is easy to
verify that the overall procedure, which we call $\mbox{computeCover}(i,Q,t)$,
successfully achieves properties (a) and (b). We have the following observation whose second
claim easily follows from the manner in which we pick our non-redundant working set.   

\begin{observation}
\label{obs:partial1}
The partial cover computed by $\mbox{computeCover}(i,Q,t)$ has the property that any
$i$-wedge with apex on $\C_i(k)$ has at most $2t$ colored points. Furthermore, if
$q$ and $q'$ are points in $Q$ such that $q \in W_i(q')$ (that is, $I(q)$ properly contains
$I(q')$), then $q'$ is colored only if $q$ is colored.
\end{observation}

\begin{figure}
\begin{center}
\input{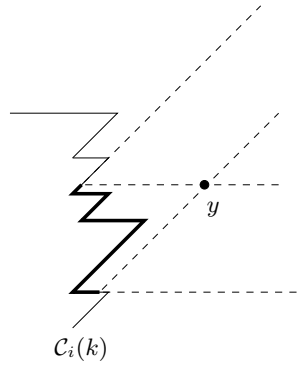}
\caption{An example of an interval $I(y)$ (in bold).  Note that the $i$-wedges with apex on $\C_i(k)$ that contain $y$ are the dotted wedges and all wedges with apex ``in between'' the apices of the dotted wedges.}
\label{fig:interval}
\end{center}
\end{figure}

\section{Decomposing Covers: Algorithm}
\label{sec:decompose}
To prove Theorem \ref{thm:decompose-wedges}, we present an algorithm that produces the
desired coloring. This is Algorithm \ref{alg:main}, but we need to define a notion before
it is fully specified.

Consider the natural linear ordering of the lines parallel to side $p_ip_{i+1}$ of $P$ (addition taken modulo $n$) with the line through vertices $p_i$ and $p_{i+1}$ being smaller than the line through any of the other vertices of $P$.  For $x,y \in \Re^2$, we define the partial order $<_i$ such that $x <_i y$ if the $p_ip_{i+1}$ parallel line through $x$ is less than the $p_ip_{i+1}$ parallel line through $y$.  

For a vertex $p_i$, let $A_i$ denote the set of all indices $j$ such that the intersection of $P$ with the line parallel to $p_jp_{j+1}$ and through $p_i$ is only the point $p_i$.  See Figure \ref{fig:Ai} for an example.

\begin{figure}
\begin{center}
\input{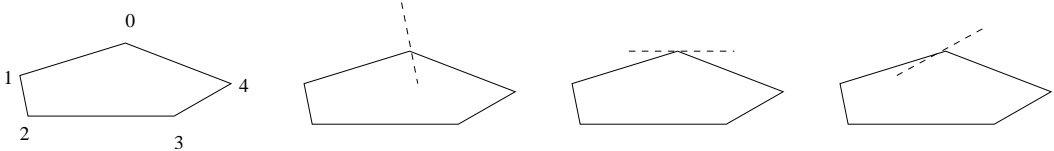}
\caption{An example of a set $A_i$.  In this example, $A_0 = \{2\}$ since only the side of $P$ parallel with $p_2p_3$ has the qualifying property.}
\label{fig:Ai}
\end{center}
\end{figure}

\begin{algorithm*}[ht]
	\caption{}
	\begin{algorithmic}[1]

	\STATE $Y' \leftarrow Y$
	\FOR{each $i \in \{0,1,2,\ldots,n-1\}$}
	   \STATE $L \leftarrow \text{min} \{\load_{Y'}(W_z(x)) : x \in \C_z(k)$ and $z = i, i+1, \ldots, n-1\}$
           %\STATE $X_i \leftarrow Y'$
	   \FOR{each $c \in \C_i(k)$}
	      \FOR{each $j \in A_i$}
		  \STATE Let $X^j(c)$ be the first $\frac{L}{2n}$ points in $W_i(c) \cap Y'$ in decreasing order with respect to the ordering $<_j$.
	      \ENDFOR
		\STATE $X(c) \leftarrow \{ W_i(c) \cap Y' \} \setminus \bigcup_{j \in A_i} X^j(c)$
	   \ENDFOR
	   \STATE \label{line:xi} $X_i \leftarrow \bigcup_{c \in \C_i(k)} X(c)$
	   \STATE Run computeCover($i$, $X_i$, $\frac{L}{64n}$).
	   \STATE Let $Y'$ denote the unscheduled points.
        \ENDFOR
	\end{algorithmic}
\label{alg:main}
\end{algorithm*}

Algorithm \ref{alg:main} calls $\mbox{computeCover}(i,X_i,t_i)$
for each $0 \leq i \leq n-1$. The set $X_i$ in the $i$-th iteration is an appropriately chosen 
subset of the points in $Y$ not colored in iterations $0,1, \ldots, i-1$. At the
beginning of the $i$-th iteration, let $L$ denote, as in the algorithm, 
the smallest number of uncolored points in a $j$-wedge with apex on $\C_j(k)$, for
$i \leq j \leq n - 1$.
The parameter $t_i$ is chosen to be $\frac{L}{64n}$, and we have 
$\min_{c \in \C_i(k)} |W_i(c) \cap X_i| \geq \frac{L}{2}$ (due to the manner in which
$X_i$ is chosen in the algorithm, see Observation 
\ref{obs:load} below). After the call to
 $\mbox{computeCover}(i,X_i,\frac{L}{64n})$, any $i$-wedge with apex on $\C_i(k)$
contains points colored $1, 2,\ldots, L/64n$.
Thus, the algorithm  produces a coloring as required in Theorem \ref{thm:decompose-wedges},
provided $L \in \Omega(k)$. This is established by the following key Lemma. It states that 
$L$, which equals $k$ before the $0$-th iteration, drops by a factor of at most $5n$
with each iteration.

\begin{lemma}
 Suppose at the beginning of iteration $i$, all $j$-wedges with apex on $\C_j(k)$ have load at least $L$ from points in $Y'$ for $j \geq i$, where $L$ is larger than some absolute constant.
(Note that $Y'$ always denotes the uncolored points in the algorithm.) After the $i$-th 
iteration of the algorithm, any $j$-wedge $W_j(x)$, for $j > i$, and with apex $x$ on 
$\C_j(k)$ has load at 
least $\frac{L}{5n}$ from points in $Y'$.
\end{lemma}

\begin{proof}
 There are two main cases to consider:

\begin{itemize}
 \item $p_i$ and $p_j$ are antipodal vertices of $\bar{P}$ -- that is, there are parallel lines
       through $p_i$ and $p_j$ with the convex polygon sandwiched between them.
 \item $p_i$ and $p_j$ are not antipodal vertices of $P$.
\end{itemize}

We use the following terminology for
iteration $i$: if for two distinct points $q$ and $q'$, 
if $W_i(q) \subseteq W_i(q')$, we say that $q$ {\em dominates} $q'$.
Notice that if $q$ and $q'$ are both unscheduled before iteration $i$,
then $q'$ is scheduled in iteration $i$ only if $q$ is already scheduled. (This
is Observation \ref{obs:partial1}.)
For the rest of this proof, let $Y'$ denote the points that are
not scheduled just before iteration $i$, let $X_i$ denote set of candidate points that are eligible to be scheduled in iteration $i$ (as constructed in the algorithm), and let $Y_i$ denote
the points that are actually scheduled in iteration $i$. 

The analysis will rely heavily upon the following two observations.

\begin{observation}
\label{obs:overcover}
For any $z \in \C_i(k)$, we have that $\load_{Y_i}(W_i(z)) \leq \frac{L}{32n}$.   
\end{observation}

\begin{observation}
\label{obs:load}
For any $z \in \C_i(k)$, $\load_{X_i}(W_i(z)) \geq \frac{L}{2}$.
\end{observation}

Observation \ref{obs:overcover} is a consequence of Observation \ref{obs:partial1}.  
To see Observation \ref{obs:load}, note that for each $c \in \C_i(k)$, $\load_{X_i}(W_i(c)) \geq \load_{X(c)}(W_i(c)) \geq L - n\cdot \frac{L}{2n} = \frac{L}{2}$.

%\ref{obs:partial}

{\noindent {\bf Case 1:} $p_i$ and $p_j$ are not antipodal vertices of $P$.}

Let $W_j(x)$ be as in the statement of the Lemma. The argument is trivial if 
$W_j(x) \cap \C_i(k) = \emptyset$. So let us assume
that $W_j(x) \cap \C_i(k) \neq \emptyset$.  
There are two cases -- in the first, we encounter $p_j$ after $p_i$ and before the vertices
antipodal to $p_i$ when walking counter-clockwise around $\bar{P}$, and in the second,
we encounter $p_j$ after the vertices antipodal to $p_i$ and before $p_i$. We will focus
on the first case, since the other is symmetric.
Let $z$ be the intersection point
of the boundary of $W_j(x)$ and $\C_i(k)$.  
If $W_j(x)$ does not contain in its interior the tail $\tail_i$ of the level 
curve $\C_i(k)$, then $W_j(x) \cap Y_i \subseteq W_i(z) \cap Y_i$, and so
$\load_{Y_i}(W_j(x)) \leq \load_{Y_i}(W_i(z)) \leq \frac{L}{32n}$.
It follows that the load of unscheduled points in $W_j(x)$ after
iteration $i$ is at least
\[ L - \frac{L}{32n} > \frac{L}{5n}. \]

Let us therefore assume that $W_j(x)$ does contain in its interior the tail
$\tail_i$ of $\C_i(k)$. See Figure \ref{fig:m3}. Let $a$ denote the 
point where the boundaries of the wedges $W_i(z)$ and $W_i(\tail_i)$ 
intersect. If $\load_{X_i}(W_i(a)) \geq \frac{L}{16n}$, then since
$\load_{Y_i}(W_i(a)) \leq \load_{Y_i}(W_i(\tail_i)) \leq \frac{L}{32n}$,
there are unscheduled points in $W_i(a)$ after iteration $i$. Since
any point in $W_i(a)$ dominates points in $W_j(x) \cap Y_i$ that are not
contained in $W_i(z) \cup W_i(\tail_i)$, we conclude that 
$W_j(x) \cap Y_i \subseteq (W_i(z) \cup W_i(\tail_i)) \cap Y_i$. Thus,
\[\load_{Y_i}(W_j(x)) \leq \load_{Y_i}(W_i(z)) + \load_{Y_i}(W_i(\tail_i))
\leq \frac{L}{16n}.\]

\noindent Therefore there must be at least $L - \frac{L}{16n} > \frac{L}{5n}$ unscheduled points left in $W_j(x)$.

\begin{figure}
\begin{center}
\input{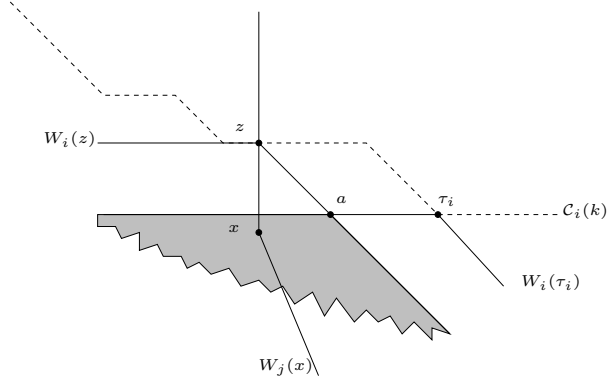}
\caption{Illustration for the nonantipodal case.}
\label{fig:m3}
\end{center}
\end{figure}

Let us therefore consider the case where 
$\load_{X_i}(W_i(a)) < \frac{L}{16n}$. This means that
$\load_{X_i}(W_i(\tail_i) \setminus W_i(a)) > \frac{L}{2} - \frac{L}{16n} > \frac{L}{3}$.  Again,
$\load_{Y_i}(W_i(\tail_i) \setminus W_i(a)) \leq \load_{Y_i}(W_i(\tail_i)) 
\leq \frac{L}{32n}$, and this means the load of the points in
$W_i(\tail_i) \setminus W_i(a)$ that are unscheduled after iteration
$i$ is at least $\frac{L}{3} - \frac{L}{32n}
> \frac{L}{5n}$. But $W_i(\tail_i) \setminus W_i(a) \subseteq
W_j(x)$, and this means that the load of the unscheduled points in $W_j(x)$
after iteration $i$ is at least $\frac{L}{5n}$.

{\noindent {\bf Case 2:} $p_i$ and $p_j$ are antipodal vertices of $P$.}

Again, the argument is trivial if $W_j(x) \cap \C_i(k) = \emptyset$. So let us assume
that $W_j(x) \cap \C_i(k) \neq \emptyset$.  Since $\load_{Y'}(W_j(x)) \geq L$, if $\load_{X_i}(W_j(x)) \leq \frac{L}{2}$ then $W_j(x)$ will clearly have load at least $\frac{L}{5n}$ after iteration $i$.  So assume that $\load_{X_i}(W_j(x)) > \frac{L}{2}$.

Consider the line parallel with $p_{i-1}p_i$ through $x$ and the line parallel with $p_ip_{i+1}$ through $x$.  Note these lines are parallel with the boundaries of an $i$-wedge.  Let $H_t(x)$ denote the halfplane consisting of all points $y$ such that $y \leq_t x$.  Let $W_j^1(x) = H_{i-1}(x) \cap H_i(x) \cap W_j(x)$.  Let $W_j^2(x) = (H_{i}(x) \cap W_j(x)) \setminus W_j^1(x)$.  Let $W_j^3(x) = (H_{i-1}(x) \cap W_j(x)) \setminus W_j^1(x)$.  See Figure \ref{fig:antipodalWedges} and Figure \ref{fig:antipodalWedges2} for an illustration.  Note that $W_j^1(x), W_j^2(x),$ and $W_j^3(x)$ form a partition of $W_j(x)$.  Also note that $W_j^1(x)$ cannot be empty but $W_j^2(x)$ or $W_j^3(x)$ could be empty.  Since these three sets form a partition of $W_j(x)$ and $\load_{X_i}(W_j(x)) > \frac{L}{2}$, it must be that one of the three sets has load at least $\frac{L}{6}$ from $X_i$.  We first handle the case when $\load_{X_i}(W_j^1(x)) \geq \frac{L}{6}$ and then conclude the proof with the case when $\load_{X_i}(W_j^2(x)) \geq \frac{L}{6}$.  The case when $\load_{X_i}(W_j^3(x)) \geq \frac{L}{6}$ has a symmetric proof with the $W_j^2(x)$ case.

\begin{figure}[htpb]
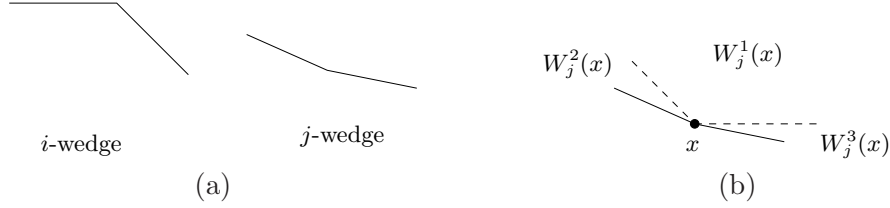

\centering
\begin{tabular}{c@{\hspace{0.1\linewidth}}c}
\input{antipodalWedges1.pstex_t} &
\input{antipodalWedges2.pstex_t} \\
(a) & (b)
\end{tabular}
\caption{If we are working with the corresponding $i$-wedge and $j$-wedge (part (a)), then we obtain the corresponding $W_j^1(x)$, $W_j^2(x)$, and $W_j^3(x)$ (part (b)).}
\label{fig:antipodalWedges}
\end{figure}

\begin{figure}[htpb]
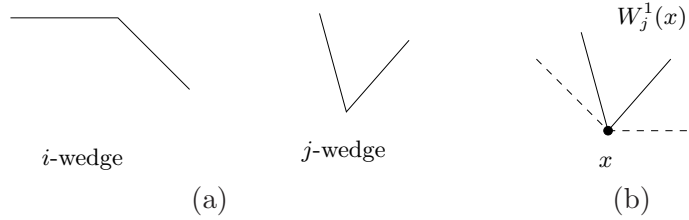

\centering
\begin{tabular}{c@{\hspace{0.1\linewidth}}c}
\input{antipodalWedges3.pstex_t} &
\input{antipodalWedges4.pstex_t} \\
(a) & (b)
\end{tabular}
\caption{If we are working with the corresponding $i$-wedge and $j$-wedge (part (a)), then $W_j^1(x) = W_j(x)$ and $W_j^2(x) = W_j^3(x) = \emptyset$ (part (b)).}
\label{fig:antipodalWedges2}
\end{figure}

{\noindent {\bf Case 2(a):} $\load_{X_i}(W_j^1(x)) \geq \frac{L}{6}$.}

Consider any point $z \in W_j^1(x) \cap \C_i(k)$.  Let $a_z$ denote the 
``leftmost'' point where the boundaries of $W_j^1(x)$ and $W_i(z)$ intersect, 
and let $b_z$ denote the ``rightmost'' point where the boundaries of $W_j^1(x)$ 
and $W_i(z)$ intersect.  Let $R_z$ be the quadrilateral with vertices 
$a_z, x, b_z,$ and $z$. That is, $R_z = W_j^1(x) \cap W_i(z)$.  Suppose that $\load_{X_i}(R_z) \geq  \frac{L}{5n} + \frac{L}{32n}$. 
Again, since $\load_{Y_i}(W_i(z)) \leq \frac{L}{32n}$, and
all points in $R_z$ are in $W_i(z)$, $R_z$ contains (unscheduled) load at least
$\frac{L}{5n}$ after iteration $i$.
Since $R_z \subseteq W_j^1(x)$, $W_j^1(x)$ contains (unscheduled) load at least
$\frac{L}{5n}$ after iteration $i$, and we are done.
See Figure \ref{fig3} for an illustration.

\begin{figure}
\begin{center}
\input{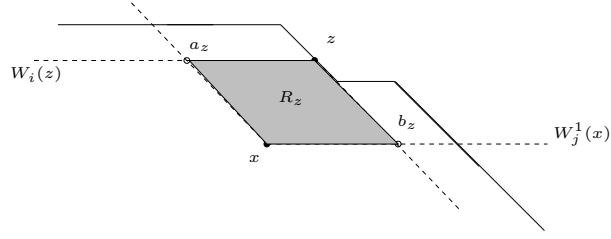}
\caption{Illustration for Case 2(a): the region $R_z$.  Note that although this figure is drawn with respect to a scenario as in Figure \ref{fig:antipodalWedges}, the analysis still holds for the scenario as in Figure \ref{fig:antipodalWedges2} (i.e. when the boundaries of an $i$-wedge are not parallel with the boundaries of $W_j^1(x)$.}
\label{fig3}
\end{center}
\end{figure}

So we can assume that 
$\load_{X_i}(R_z) \leq  \frac{L}{5n} + \frac{L}{32n}$ for each 
$z \in W_j^1(x) \cap \C_i(k)$.  Since $\load_{X_i}(W_i(z)) \geq \frac{L}{2}$, we must have
$\load_{X_i}( W_i(a_z) \cup W_i(b_z)) \geq
\frac{L}{2} - (\frac{L}{5n} + \frac{L}{32n}) > \frac{L}{8}$.  Let $z_1$ be the ``leftmost'' point on $\C_i(k) \cap W_j^1(x)$, 
and let $z_2$ be the ``rightmost'' point on $\C_i(k) \cap W_j^1(x)$.
Notice that $a_{z_1}$ is just $z_1$ itself, and so
$\load_{X_i}(W_i(a_{z_1})) \geq \frac{L}{2}$.  Similarly, $\load_{X_i}(W_i(b_{z_2})) \geq \frac{L}{2}$.
Let $z'$ be the last point on $\C_i(k)$, while walking from $z_1$
to $z_2$, such that $\load_{X_i}(W_i(a_{z'})) \geq \frac{L}{16}$. Thus
\[\load_{X_i}(W_i(b_{z'})) \geq \load_{X_i}(W_i(a_{z'}) \cup W_i(b_{z'})) - \frac{L}{16} \geq \frac{L}{8} -  \frac{L}{16} = \frac{L}{16}.\]
See Figure \ref{fig4} for an illustration.

\begin{figure}
\begin{center}
\input{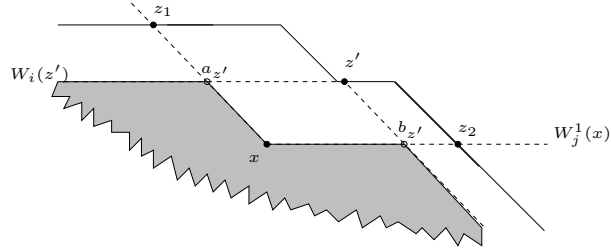}
\caption{Illustration for Case 2(a): the constructed point $z'$.}
\label{fig4}
\end{center}
\end{figure}

Now consider any point $z'' \in W_j^1(x) \setminus W_i(z')$.  
It must be that $W_i(z'')$ either contains $W_i(a_{z'})$ or 
contains $W_i(b_{z'})$ which both have load in $X_i$ of at 
least $\frac{L}{16}$.  Suppose that $W_i(z'')$ contains 
$W_i(a_{z'})$; the other case is similar.
The points in $W_i(a_{z'})$ dominate $z''$ and 
we will not schedule $z''$ in iteration $i$ until we have scheduled all points
in $W_i(a_{z'}) \cap X_i$.
But since $\load_{Y_i}(W_i(a_{z'})) \leq \frac{L}{32n} < \frac{L}{16} \leq
\load_{X_i}(W_i(a_{z'}))$, this means we will not schedule $z''$.

It follows that $W_j^1(x) \cap Y_i \subseteq W_i(z') \cap Y_i$, and
thus $\load_{Y_i}(W_j^1(x)) \leq  \load_{Y_i}(W_i(z')) \leq \frac{L}{32n}$.
And so the load of unscheduled points in $W_j^1(x)$ after iteration
$i$ is at least 
$L - \frac{L}{32n} \geq \frac{L}{5n}$.

{\noindent {\bf Case 2(b):} $\load_{X_i}(W_j^2(x)) \geq \frac{L}{6}$.}

If $W_j^2(x) \cap \C_i(k) = \emptyset$, then the lemma trivially holds.  So for now on, we will assume $W_j^2(x) \cap \C_i(k) \neq \emptyset$.  Let $z \in \C_i(k)$ be a point such that $W_i(z) \cap W_j^2(x) \neq \emptyset$.  Note that both $W_j^2(x)$ and $W_i(z)$ have a boundary parallel with the side $p_{i-1}p_i$.  There are only two types of intersections between these two wedges:

\begin{enumerate}
  \item $z \in W_j^2(x)$, the boundary of $W_j^2(x)$ parallel with the side $p_{j-1}p_{j}$ intersects both boundaries of $W_i(z)$, and the boundary of $W_j^2(x)$ parallel with the side $p_{i-1}p_{i}$ does not intersect with $W_i(z)$.
  \item $x \in W_i(z)$, the boundary of $W_i(z)$ parallel with the side $p_ip_{i+1}$ intersects both boundaries of $W_j^2(x)$, and the boundary of $W_i(z)$ parallel with the side $p_{i-1}p_{i}$ does not intersect with $W_j^2(x)$.
\end{enumerate}

See Figure \ref{fig:intersect1} for an illustration.

\begin{figure}[htpb]
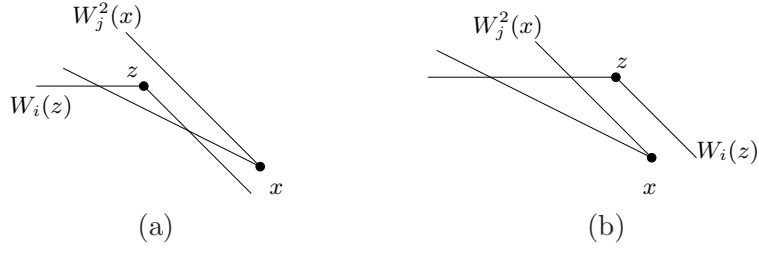

\centering
\begin{tabular}{c@{\hspace{0.1\linewidth}}c}
\input{type3.pstex_t} &
\input{type4.pstex_t} \\
(a) & (b)
\end{tabular}
\caption{Illustration for case 2(b): (a) A type 1 intersection.
(b) A type 2 intersection.}
\label{fig:intersect1}
\end{figure}

Let $\{v_1, v_2, v_3, \ldots \}$ denote the points in $X_i \cap W_j^2(x)$ in decreasing order according to the partial order $<_i$.  Let 
$\ell = \textrm{max}\{t | v_t \in Y_i\}$.  If $\ell \leq \frac{L}{10}$ then $\load_{Y'}(W_j(x)) \geq \frac{L}{6} - \frac{L}{10} > \frac{L}{5n}$ after iteration $i$, so assume that $\ell > \frac{L}{10}$.

Since $v_\ell \in X_i$ there is a $u \in \C_i(k)$ so that $v_{\ell} \in X(u)$ in iteration $i$ of
the algorithm. Suppose that the intersection between $W_i(u)$ and $W_j^2(x)$ is a type 1 intersection.  Let $T_{v_\ell} = W_i(u) \setminus H_{j-1}(v_\ell)$.  Note that since $W_j^2(x) \neq \emptyset$, it must be that $j-1 \in A_i$.  Combining this with the fact that $v_\ell \in X(u)$, we 
know that $X^{j-1}(u) \subset T_{v_\ell} \subset W_i(u)$.  (See Algorithm for the notation.) 
Thus there are at least $|X^{j-1}(u)| = \frac{L}{2n}$ points from $Y'$ in $T_{v_\ell}$.  Since, $\load_{Y_i}(W_i(u)) \leq \frac{L}{32n}$, there must be at least $\frac{L}{2n} - \frac{L}{32n} > \frac{L}{5n}$ unscheduled points left in $T_{v_\ell}$ after iteration $i$.  Since we are dealing with a type 1 intersection, $T_{v_\ell} \subset W_j^2(x)$, and thus $W_j^2(x)$ will contain at least $\frac{L}{5n}$ unscheduled points after iteration $i$ and the lemma holds.  See Figure \ref{fig:heavy} for an illustration.

Now suppose that the intersection between $W_i(u)$ and $W_j^2(x)$ is a type 2 intersection.  Consider the region $T'_{v_\ell} = W_j^2(x) \setminus H_i(v_\ell)$.  Since we are assuming $\ell > \frac{L}{10}$, it must be that $\load_{Y'}(T'_{v_\ell}) \geq \frac{L}{10}$.  Since we are dealing with a type 2 intersection, it must be that $T'_{v_\ell} \subset W_i(u)$.  Since $\load_{Y_i}(W_i(u)) \leq \frac{L}{32n}$, we have that $\load_{Y_i}(T'_{v_\ell}) \leq \frac{L}{32n}$ and thus there will be at least $\frac{L}{10} - \frac{L}{32n} > \frac{L}{5n}$ unscheduled points left in $T'_{v_\ell}$ after iteration $i$.  Since $T'_{v_\ell} \subseteq W_j^2(x)$, there must be 
(unscheduled) load at least $\frac{L}{5n}$ in $W_j^2(x)$ after iteration $i$.  See Figure \ref{fig:heavy} for an illustration.

\begin{figure}[htpb]
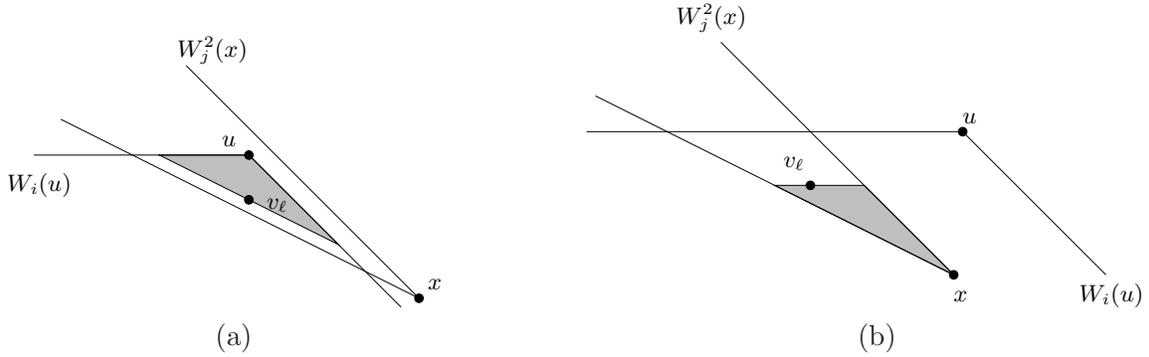

\centering
\begin{tabular}{c@{\hspace{0.1\linewidth}}c}
\input{r1.pstex_t} &
\input{r2.pstex_t} \\
(a) & (b)
\end{tabular}
\caption{Illustration for case 2(b): (a) An illustration of $T_{v_\ell}$.
(b) An illustration of $T'_{v_\ell}$.}
\label{fig:heavy}
\end{figure}

\end{proof}

\section{Restricted Strip Cover}
\label{sec:RSC}
We now describe our algorithm for the RSC which takes an instance consisting of 
sensors $\inst$ and universe $U$ and returns a schedule $S$ of the sensors.  
We will later show that the schedule produced by the algorithm has duration at least $L/5$.  
The reader way wish to consult Section \ref{sec:intro} for the notation used in the
context of RSC. The algorithm starts with
the empty schedule where no sensor is assigned, and assigns a start time
to one sensor in each iteration. We will also denote the current schedule
of the algorithm at any stage of its execution by $S$.

With respect to a schedule $S$, we say the sensor $s$ \textit{dominates coordinate $x$ to the right} if $s$ extends as far to the right as possible
(maximizes $r(s)$) among all sensors that have not been assigned and are
live at $x$. In the event of a tie, we take the sensor that extends as far to 
the left as possible. Further ties are broken arbitrarily.  The sensor that \textit{dominates coordinate $x$ to the left} is defined similarly.
For ease of description, define $\dur(S,0) = \dur(S,m+1) = \infty$.

%\pagebreak
\begin{algorithm*}[ht]
	\caption{}
	\begin{algorithmic}[1]

	\STATE $t \leftarrow 0$
	\STATE $S \leftarrow \emptyset$\\

	\WHILE{TRUE}
		\STATE $t \leftarrow \dur(S) + 1$
		\STATE Let $i$ be the first uncovered coordinate at time $t$ and let $j$ be max $\{x \in U$ $|$ $[i,x]$ is uncovered at time $t \}$.
		\STATE Let $s'$ be the sensor that dominates $i$ to the right.  If $s'$ does not exist, go to step 9.
		\STATE If $s'$ is not live at $j$, $t(s') \leftarrow t$ and $S \leftarrow S \cup \{s'\}$.
		\STATE If $s'$ is live at $j$, let $s''$ be the sensor that dominates $j$ to the left.  If $\dur(S,i-1) \geq \dur(S,j+1)$, $t(s') \leftarrow t$ and $S \leftarrow S \cup \{s'\}$.  Otherwise, $t(s'') \leftarrow t$ and $S \leftarrow S \cup \{s''\}$.
	\ENDWHILE

	\STATE Return $S$.
	\end{algorithmic}
\end{algorithm*}

Let us denote by $t_f$ the duration of the
final schedule produced by the algorithm. At termination, there
is a point $x \in U$ so that $\dur(S,x) = t_f$ and there is no
unassigned sensor that is live at $x$.

\paragraph{Running Time.}
We will iterate through the while loop at most $n$ times because we schedule a sensor in each iteration of the while loop.
Each iteration of the while is readily implemented in
$O(m + n) = O(n)$ time. Thus the algorithm runs in $O(n^2)$ time.  
It may be possible to improve the running time by using data structures
such as segment trees.

\subsection{Approximation Ratio}

\label{analysis}
For an instance $\inst$ of RSC, let $OPT$ denote the duration of an optimal solution for 
$\inst$. We have the following theorem, whose proof is in the appendix.

\begin{theorem}
Given any instance $\inst$ of Restricted Strip Cover, our algorithm returns
a schedule $S$ such that $\dur(S) \geq OPT/5$.
\label{constant}
\end{theorem}

\subsection{Implication for Planar Sensor Cover}

Here we generalize Theorem \ref{thm:decompose-main} to show that given an instance
of the planar sensor cover problem, we can compute a schedule whose duration is
within a constant factor of the load. The proof proceeds along similar lines, and
we only need an analog in the non-uniform duration case 
of the procedure $\mbox{computeCover}()$ in Section \ref{sec:simple}. Now our algorithm
for RSC furnishes just such an analog, with Lemma \ref{subset} and Observation \ref{obs:partial}
being the equivalent of Observation \ref{obs:partial1}. The rest of the proof follows from
a careful but fairly mechanical modification to the proof of Theorem \ref{thm:decompose-main}.
We conclude with a statement of our constant factor approximation for planar sensor
cover:

\begin{theorem}
For any given convex polygon, there is a polynomial time algorithm
for the planar sensor cover problem (with ranges being translates of the convex
polygon) that computes a schedule whose duration is within a constant factor
of the load of the instance.
\end{theorem}

\bibliography{sensorbib}

\appendix

\section{Proof of Theorem \ref{constant}}

\begin{lemma}
\label{subset}
 Given some instance of Restricted Strip Cover $\inst$, let $S$ be the schedule returned by our algorithm.  
Let $u,v \in \inst$ be any two, distinct sensors.  If $R(u)$ is strictly
contained in $R(v)$, then $u$ is scheduled after $v$ and
in fact $t(u) \geq t(v) + d(v)$.
\end{lemma}

\begin{proof}

 Suppose that $u$ and $v$ are two unscheduled sensors such that $R(u)$
is strictly contained in  $R(v)$.  Sensors are only scheduled when they 
dominate some coordinate to the left or to the right.  Suppose we want to find 
the sensor that dominates some coordinate $i \in [\ell(u), r(u)]$.  We will consider both $u$ and $v$, but will always prefer $v$ to $u$ from the definition of domination.  Therefore, we will schedule $v$ before $u$ and will not consider
another sensor to dominate a coordinate in $[\ell(u), r(u)]$ until after time $t(v) + d(v)$.
\end{proof}

For $x \in U$ and $t > 0$, we define $\cover(x,t)$ to be the number 
of sensors that cover $x$ at time $t$ in the schedule output by
our algorithm.

We need the following observation.

\begin{lemma}

If $\cover(x,t) \leq c$ for each $x \in U$ and $t > 0$, then the duration
$t_f$ of the schedule we output is at least $L/c$.

\label{coverage}

\end{lemma}

\begin{proof}

At termination, there is a point $x \in U$ so that $\dur(S,x) = t_f$ and there 
is no unassigned sensor that is live at $x$. Thus, 
$c t_f \geq L(x) \geq L$, and so $t_f \geq L/c$. 
\end{proof}

We will now show that $\cover(x,t) \leq 5$ for each $x \in U$ and $t > 0$.
Theorem \ref{constant} then follows immediately from this and Lemma \ref{coverage}.

In each iteration the algorithm schedules (assigns) a sensor $s$ which
is either $s'$ or $s''$. Let us call the corresponding interval $[i,j]$ 
the {\em interval for which $s$ is scheduled}. (Please refer to the 
algorithm for what $i$ and $j$ stand for.) If $s = s'$, we call $s$ a 
{\em right going} sensor to remember that it was chosen to dominate $i$ to the 
right. In this case, $i$ was not covered at time $t(s)$ before $s$ was
scheduled, but $i$ was covered at time $t(s)$ after $s$ was
scheduled. We say that $s = s'$ {\em closes} $i$ at time $t(s)$.
Similarly, if $s = s''$ we call $s$ a {\em left going} sensor and say
that it closes $j$ at time $t(s)$.

\begin{lemma}

For any $x \in U$ and $t > 0$, $\cover(x,t) \leq 5$.

\label{cover5}

\end{lemma}

\begin{proof}

Fix an $x \in U$ and $t > 0$. If no sensor in the output schedule
covers $(x,t)$, then $\cover(x,t)$ is $0$. Let us therefore suppose
that some sensor covers $(x,t)$, and let $s_0$ denote the first scheduled
sensor that  covers $(x,t)$.  Let us classify any other sensor $s$ that covers
$(x,t)$ into exactly one of the following four types: (1) $s$ closes some
$i < x$ and is left going; (2) $s$ closes some $i < x$ and is
right going; (3) $s$ closes some $i > x$ and is left going; 
(4) $s$ closes some $i > x$ and is right going. We show that there
are at most two sensors of types 1 and 2 put together.
A symmetric argument shows that there are at most two sensors of types 
3 and 4 put together.

\begin{figure}
\begin{center}

\includegraphics[width=.5\linewidth]{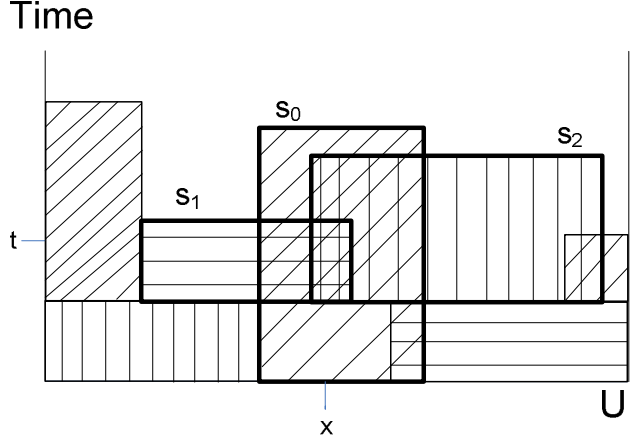}

\caption{There are 3 sensors covering $(x, t)$.  The first scheduled sensor to cover $(x,t)$ is $s_0$.  The next sensor to cover $(x, t)$ is a type 2 sensor $s_1$.  Finally, $(x, t)$ is covered by a type 4 sensor $s_2$.}
\label{fig:RSC}
\end{center}
\end{figure}

Suppose that at some point the algorithm adds a sensor $l$ of type 1. We
claim that after $l$ is added no more sensors of types 1 or 2 are
added.  Suppose $l$ closed $i < x$ at time $t(l)$. Consider some sensor
$l'$ that is live at $x$ and that closes some $i' < x$ when it is later
added by the algorithm. Since $l$ was chosen because it dominated 
$i$ to the left, we can conclude that $\ell(l') \geq \ell(l)$. (If this were 
not the case, then $l'$ would necessarily be live at $i$ and would have been preferred to $l$.)  Observe that the interval $[\ell(l), x]$ is covered by $l$ for all times between $t(l)$ and $t(l) + d(l) - 1$.  
Since $i' \in [\ell(l'), x] \subseteq [\ell(l),x]$, we must have
$t(l') > t(l) + d(l) - 1 \geq t$.  Thus $l'$ will not cover $x$ at time $t$.
We conclude that once we schedule a sensor of type 1, we do not schedule
any more sensors of types 1 or 2.

We will therefore consider the case where the first sensor of type 1 or 2 is a
type 2 sensor which we denote $r_1$. We need the following claim.

\begin{claim}

\label{cl:rhs}

Let $v$ be a sensor scheduled after $r_1$ such that (a) $v$ is live at
$x$, (b) $t(v) \leq t$, and (c) $v$ closes some $z < x$ at time
$t(v)$. Let $[x',y']$ be the interval for which $v$ is scheduled. 
Then $y'+1 = \ell(r_1)$.

\end{claim}

We prove the claim after completing the rest of the proof of the lemma.
If the next sensor of type 1 or 2 that we schedule after $r_1$ is
a type 1 sensor, we do not schedule any more sensors of types 1 and
2. So let us assume that the next sensor of types 1 or 2 that we schedule
after $r_1$ is a type 2 sensor $r_2$. Since $t \geq t(r_2) \geq t(r_1)$,
and $r_2$ closes some $i' < x$ at time $t(r_2)$, we must have $i'
< \ell(r_1)$, and thus $\ell(r_2) < \ell(r_1)$. Now assume for the sake of
a contradiction that there is some sensor $r_3$ of type 1 or 2 
that is scheduled after $l_2$. Reasoning as above, the interval
$[x'',y'']$ that it is scheduled for has $y'' < \ell(r_2)$. Thus
$y'' + 1 \leq \ell(r_2) < \ell(r_1)$, a contradiction to Claim \ref{cl:rhs}.

Thus we have completed the proof of the lemma and we now prove
Claim \ref{cl:rhs}.

{\bf Proof of Claim \ref{cl:rhs}.} First, we clarify to the reader
that in the statement of the claim condition (a) does not require
$v$ to cover $(x,t)$, but only requires $v$ to be live at $x$. 
Let $v_1, \ldots, v_k$ be the sensors satisfying the hypothesis
of the claim, ordered in the sequence in which they were chosen by 
the algorithm. Let $[x_j,y_j]$ be the interval for which $v_j$ 
is scheduled. Since $t(r_1) \leq t(v_j) \leq t$ and $v_j$ closes
some point strictly less than $x$, we must have $y_j < \ell(r_1)$.
We have argued that $y_j \leq \ell(r_1) - 1$, and we want to show
$y_j = \ell(r_1) - 1$. Suppose for the sake of contradiction that
$y_j < \ell(r_1) - 1$ for some $j$, and consider the first $j$ for
which this happens.

When $v_j$ was being scheduled, $y_j + 1$ is covered at
time $t(v_j)$ by a sensor $w$. Clearly, it must be 
that $y_j + 1 = \ell(w)$. So we have $\ell(v_j) < \ell(w) < \ell(r_1)$.

If $w$ was scheduled before $r_1$, then $r(w) < i' < x$ where
$i'$ is the point that $r_1$  closes. Thus $R(v_j)$ properly
contains $R(w)$ but $v_j$ is scheduled after $w$, a contradiction
to Lemma \ref{subset}.

If $w$ was scheduled after $r_1$, then $w$ must be live at $x$,
for otherwise $R(v_j)$ properly contains $R(w)$ and we derive
a contradiction as above. Also, we must have $r(w) < r(r_1)$
for otherwise $R(w)$ properly contains $R(r_1)$ and we reach
a contradiction. Since $r(w) < r(r_1)$, $w$ closes some
$i < \ell(r_1)$. Thus $w$ is scheduled after $r_1$, $t(w) \leq t(v_j)
\leq t$, $w$ is live at $x$, and $w$ closes some point strictly to the
left of $x$. Thus $w = v_{j'}$ for some $j' < j$. If $j = 1$,
we have reached a contradiction. We therefore assume $j > 1$. 
We observe that $w$ cannot be a left going sensor, because
$v_j$ is also live at the point $i$ that $w$ closes, and would
have been preferred to $w$ otherwise since $\ell(v_j) 
< \ell(w)$. Since $j' < j$, we have $y_{j'} + 1 = \ell(r_1)$.
Thus $\Dur(S,y_{j'} + 1) \geq t$ at the time $w = v_{j'}$
is being scheduled. Since $w$ is right going, we must have
$\Dur(S,x_{j'} - 1) \geq \Dur(S,y_{j'} + 1) \geq t$ at the time
$w$ is being scheduled. Let $s$ denote the sensor that covers
$x_{j'}-1$ at times $t(w), t(w) + 1, \ldots, t$ when $w$ is
being scheduled. We have $r(s) = x_{j'} - 1 < x_{j'} < x \leq r(v_j)$. Now $v_j$ closes some point $i'' \in [x_j, y_j]$
at time $t(v_j)$, and since $t \geq t(v_j) \geq t(w)$, $i''$
cannot be in $R(s)$. Since $i'' \leq y_j = \ell(w) - 1 
\leq x_{j'} - 1 = r(s)$, it must be that $i'' < \ell(s)$.
So $\ell(v_j) \leq i'' < \ell(s)$ and we have already argued that
$r(v_j) > r(s)$. So $R(v_j)$ properly contains $R(s)$ but $v_j$ is
scheduled after $s$, a contradiction to Lemma \ref{subset}.

\end{proof}

The following observation about our algorithm for RSC is evident
from the analysis.

\begin{observation}
\label{obs:partial}
Suppose that we stop our algorithm once the duration of the schedule
becomes greater than equal to $t$. Then the total load of all the
scheduled sensors that are live at some point $x$ of the universe
is at most $5 (t + d_{max})$, where $d_{max}$ is the maximum duration of
any input sensor.
\end{observation}

\end{document}